\documentclass[12pt]{article}
\usepackage[latin1]{inputenc}
\usepackage[british]{babel}
\usepackage{cmap}
\usepackage{lmodern}

\usepackage{amssymb, amsmath, amsthm}
\usepackage[a4paper,top=25mm,bottom=25mm,left=25mm,right=25mm]{geometry}
\usepackage{etex}
\usepackage{ragged2e}

\usepackage{authblk} 
\usepackage{pifont}
\usepackage{graphicx}
\usepackage[usenames,dvipsnames,svgnames,table]{xcolor}
\usepackage[figuresright]{rotating}
\usepackage{xtab} 
\usepackage{longtable} 
\usepackage{ltxtable} 
\usepackage{multirow}
\usepackage{footnote}
\usepackage[stable]{footmisc}
\usepackage{chngpage} 
\usepackage{pdflscape} 
\usepackage[nottoc,notlot,notlof]{tocbibind} 

\usepackage{pgfplots}
\pgfplotsset{compat=1.14}
\usepackage{setspace}

\usepackage{array}
\newcolumntype{K}[1]{>{\centering\arraybackslash$}p{#1}<{$}}
\usepackage{wasysym}

\makesavenoteenv{tabular}
\usepackage{tabularx}
\usepackage{booktabs}
\usepackage{threeparttable}
\usepackage[referable]{threeparttablex} 
\newcolumntype{R}{>{\raggedleft\arraybackslash}X}
\newcolumntype{L}{>{\raggedright\arraybackslash}X}
\newcolumntype{C}{>{\centering\arraybackslash}X}
\newcolumntype{M}[1]{>{\centering\arraybackslash}m{#1}}
\newcolumntype{A}{>{\columncolor{gray!25}}C}
\newcolumntype{a}{>{\columncolor{gray!25}}c}

\newlength{\tablen}

\usepackage{dcolumn} 
\newcolumntype{.}{D{.}{.}{-1}}

\usepackage{tikz}
\usetikzlibrary{arrows, automata, calc, matrix, patterns, positioning, trees}
\usepackage[semicolon]{natbib}
\usepackage[hyphens]{url}
\makeatletter
\g@addto@macro{\UrlBreaks}{\UrlOrds}
\makeatother
\usepackage{hyperref} 
\hypersetup{
  colorlinks   = true,    		
  urlcolor     = blue,    		
  linkcolor    = blue,    		
  citecolor    = ForestGreen  	
}
\usepackage{microtype}
\usepackage[justification=centerfirst]{caption} 

\usepackage[labelformat=simple]{subcaption}

\DeclareCaptionLabelFormat{parenthesis}{(#2)}
\makeatletter
\renewcommand\p@subfigure{\arabic{figure}.}
\makeatother

\DeclareCaptionLabelFormat{parenthesis}{(#2)}
\makeatletter
\renewcommand\p@subtable{\arabic{table}.}
\makeatother

\def\addlegendimage{\csname pgfplots@addlegendimage\endcsname}

\usepackage{enumitem}
\setlist[itemize]{leftmargin=2.5\parindent}
\setlist[enumerate]{leftmargin=2.5\parindent}

\theoremstyle{plain}

\newtheorem{theorem}{Theorem}[section]

\theoremstyle{definition}
\newtheorem{axiom}{Axiom}

\newtheorem{definition}{Definition}[section]
\newtheorem{example}{Example}[section]

\theoremstyle{remark}
\newtheorem{notation}{Notation}[section]
\newtheorem{remark}{Remark}[section]

\def\keywords{\vspace{.5em} 
{\noindent \textit{Keywords}: }}

\def\JEL{\vspace{.5em} 
{\noindent \textbf{\emph{JEL} classification number}: }}

\def\AMS{\vspace{.5em} 
{\noindent \textbf{\emph{MSC} class}: }}

\author{\href{https://sites.google.com/site/laszlocsato87}{L\'aszl\'o Csat\'o}\thanks{~E-mail: \emph{csato.laszlo@sztaki.mta.hu}} }
\affil{Institute for Computer Science and Control, Hungarian Academy of Sciences (MTA SZTAKI) \\
Laboratory on Engineering and Management Intelligence, Research Group of Operations Research and Decision Systems}
\affil{Corvinus University of Budapest (BCE) \\
Department of Operations Research and Actuarial Sciences}
\affil{Budapest, Hungary}

\title{Journal ranking should depend \\ on the level of aggregation}
\date{\today}

\def\Dedication{
{\noindent
Ich behaupte aber, da{\ss } in jeder besonderen Naturlehre nur so viel eigentliche Wissenschaft angetroffen werden k\"onne, als darin Mathematik anzutreffen ist.}\footnote{~``\emph{I maintain that in each particular natural science there is only as much true science as there is mathematics.}'' (Source: Smith, J. T.: David Hilbert's 1930 Radio Address -- German and English. \url{https://www.maa.org/book/export/html/326610})}
\flushright
\noindent (Immanuel Kant: \emph{Metaphysische Anfangsgr\"unde der Naturwissenschaft})

\vspace{0.5cm} 
\justify }

\begin{document}

\maketitle

\Dedication

\begin{abstract}
\noindent
Journal ranking is becoming more important in assessing the quality of academic research. Several indices have been suggested for this purpose, typically on the basis of a citation graph between the journals. We follow an axiomatic approach and find an impossibility theorem: any self-consistent ranking method, which satisfies a natural monotonicity property, should depend on the level of aggregation. Our result presents a trade-off between two axiomatic properties and reveals a dilemma of aggregation.


\AMS{91A80, 91B14}

\JEL{C44, D71}
\end{abstract}

\section{Introduction}

The measurement of the quality and quantity of academic research plays an increasing role in the evaluation of researchers and research proposals. This paper will focus on a particular field of scientometrics, that is, journal ranking.
Furthermore, since a number of bibliometric indices have been suggested to assess intellectual influence, and now a plethora of ranking methods are available to measure the performance of journals and scholars \citep{Palacios-HuertaVolij2014}, we follow an axiomatic approach because the introduction of some reasonable axioms or conditions is able to narrow the set of appropriate methods, to reveal their crucial properties, and to allow for their comparison.

An important contribution of similar analyses can be an axiomatic characterisation, meaning that a set of properties uniquely determine a preference vector. For example, \citet{Palacios-HuertaVolij2004} give a characterisation of the invariant method, while \citet{Demange2014} provides a characterisation of the handicap method, both of them used to rank journals.
Results for citation indices are probably even more abundant, including characterisations of the $h$-index \citep{Kongo2014, Marchant2009, Miroiu2013, Quesada2010, Quesada2011a, Quesada2011b, Woeginger2008a}, the $g$-index \citep{Woeginger2008b, Quesada2011a, AdachiKongo2015}, the Euclidean index \citep{PerryReny2016}, or a class of step-based indices \citep{ChambersMiller2014}, among others. \citet{delaVegaVolij2018} characterise scholar rankings admitting a measure representation.
There are also axiomatic comparisons of bibliometric indices \citep{BouyssouMarchant2014, BouyssouMarchant2016}.

However, the above works seldom uncover the inevitable trade-offs between different natural requirements, an aim which can be achieved mainly by impossibility theorems. Similar results are well-established in social choice theory since Arrow's impossibility theorem \citep{Arrow1951} and the Gibbard-Satterthwaite theorem \citep{Gibbard1973, Satterthwaite1975, DugganSchwartz2000} but not so widely used in scientometrics.

We provide an impossibility result in journal ranking. In particular, it will be proved that two axioms, invariance to aggregation and self-consistency, cannot be satisfied simultaneously even on a substantially restricted domain of citation graphs. Invariance to aggregation means that the ranking of two journals is not influenced by the level of aggregation among the remaining journals, while self-consistency, introduced by \citet{ChebotarevShamis1997a}, is a kind of monotonicity property, responsible for some impossibility theorems in ranking from paired comparisons \citep{Csato2019d, Csato2019e}.

The paper is organised as follows.
Our setting and notations are introduced in Section~\ref{Sec2}. Section~\ref{Sec3} motivates and defines the two axioms, which turn out to be incompatible in Section~\ref{Sec4}. Section~\ref{Sec5} summarises the main findings and concludes.

\section{The journal ranking problem} \label{Sec2}

A journal ranking problem consists of a group of journals and their respective citation records \citep{Palacios-HuertaVolij2014}.
Let $N = \{ J_1,J_2, \dots, J_n \}$, $n \in \mathbb{N}$ be a non-empty finite \emph{set of journals} and $C = \left[ c_{ij} \right] \in \mathbb{R}^{n \times n}$ be a $|N| \times |N|$ nonnegative \emph{citation matrix} for $N$.
The entry $c_{ij}$ can be directly the number of citations that journal $J_i$ received from journal $J_j$, or any reasonable transformation of this value, for example, by using exponentially decreasing weights for older citations.

The pair $(N,C)$ is called a \emph{journal ranking problem}.
The set of journal ranking problems with $n$ journals ($|N| = n$) is denoted by $\mathcal{J}^n$.

The aim is to aggregate the opinions given in the citation matrix into a single judgement. Formally, a \emph{scoring procedure} $f$ is a $\mathcal{J}^n \to \mathbb{R}^n$ function that takes a journal ranking problem $(N,C)$ and returns a rating $f_i(N,C)$ for each journal $J_i \in N$, representing this judgement.

A scoring method immediately induces a ranking $\succeq$ for the journals of $N$ (a transitive and complete weak order on the set of $N$): $f_i(N,C) \geq f_j(N,C)$ means that journal $J_i$ is ranked weakly above $J_j$, denoted by $J_i \succeq J_j$. The symmetric and asymmetric parts of $\succeq$ are denoted by $\sim$ and $\succ$, respectively: $J_i \sim J_j$ if both $J_i \succeq J_j$ and $J_i \preceq J_j$ hold, while $J_i \succ J_j$ if $J_i \succeq J_j$ holds but $J_i \preceq J_j$ does not hold.

A journal ranking problem $(N,C)$ has the symmetric \emph{matches matrix} $M = C + C^\top = \left[ m_{ij} \right] \in \mathbb{R}^{n \times n}$ such that $m_{ij}$ is the number of the citations between the journals $J_i$ and $J_j$ in both directions, which can be called the number of matches between them in the terminology of sports \citep{KoczyStrobel2010, Csato2015a}.

It is sometimes convenient to consider not a general problem, arising from complicated networks of citations, but only a simpler one. 

A journal ranking problem $(N,C) \in \mathcal{J}^n$ is called \emph{balanced} if $\sum_{X_k \in N} m_{ik} = \sum_{X_k \in N} m_{jk}$ for all $J_i,J_j \in N$.
The set of balanced journal ranking problems is denoted by $\mathcal{J}_{B}$.
In a balanced journal ranking problem, all journals have the same number of matches.


A journal ranking problem $(N,C) \in \mathcal{J}^n$ is called \emph{unweighted} if $m_{ij} \in \{ 0; 1 \}$ for all $J_i,J_j \in N$.
The set of unweighted journal ranking problems is denoted by $\mathcal{J}_{U}$.
In an unweighted journal ranking problem, either there is no citations, or there exists only one citation between any pair of journals.

A journal ranking problem $(N,C) \in \mathcal{J}^n$ is called \emph{loopless} if $c_{ii} = 0$ for all $J_i \in N$.
The set of unweighted journal ranking problems is denoted by $\mathcal{J}_{L}$.
In a loopless problem, self-citations are disregarded.

The subsets of balanced, unweighted, and loopless journal ranking problems restrict the matches matrix $M$.

A journal ranking problem $(N,C) \in \mathcal{J}^n$ is called \emph{extremal} if $|c_{ij}| \in \{ 0; m_{ij} / 2; m_{ij} \}$ for all $J_i,J_j \in N$.
The set of extremal journal ranking problems is denoted by $\mathcal{J}_{E}$.
In an extremal journal ranking problem, only three cases are allowed in the comparison of journals $J_i$ and $J_j$: there are citations only for $J_i$ or $J_j$, or they are tied with respect to mutual citations. 

Any intersection of these special classes can be considered, too.

While a given citation matrix $C$ will seldom lead to a balanced, unweighted, loopless, or extremal journal ranking problem in practice, they can still be relevant for applications due to the possible transformation of citations. For example, it may make sense to remove self-citations from matrix $C$, and consider only three types of paired comparisons in the derived matrix $\hat{C}$:
\begin{itemize}
\item
$\hat{c}_{ij} = 0$ if $c_{ij} = 0$ and $c_{ji} = 0$;
\item
$\hat{c}_{ij} = 0$ if $c_{ji} > 0$ and $c_{ij} < c_{ji} / 2$;
\item
$\hat{c}_{ij} = 0.5$ if $c_{ji} > 0$ and $c_{ji} / 2 \leq c_{ij} \leq 2c_{ji}$;
\item
$\hat{c}_{ij} = 1$ if $2c_{ji} < c_{ij}$.
\end{itemize}
In other words, two journals are not compared ($\hat{c}_{ij} = \hat{c}_{ji} = 0$) if they do not cite each other, their paired comparison is tied ($\hat{c}_{ij} = \hat{c}_{ji} = 0.5$) if their mutual citations are approximately balanced -- that is, $J_i$ does not refer to $J_j$ more than two times than $J_j$ refers to $J_i$, and vice versa --, and $J_i$ is maximally better than $J_j$ ($\hat{c}_{ij} = 1$ and $\hat{c}_{ji} = 0$) if $J_j$ cites $J_i$ more than two times than $J_i$ cites $J_j$. Then the resulting journal ranking problem $\left( N,\hat{C} \right) \in \mathcal{J}^n$ is unweighted, loopless, and extremal.

\section{Axioms of journal ranking} \label{Sec3}

In this section two properties, a natural axiom of aggregation and a variant of monotonicity, are introduced. 

\subsection{Invariance to aggregation} \label{Sec31}

The first condition aims to regulate the ranking if two journals are aggregated into one.

\begin{axiom} \label{Axiom31}
\emph{Invariance to aggregation} ($IA$):
Let $(N,C) \in \mathcal{J}^n$ be a journal ranking problem and $J_i, J_j \in N$ be two different journals.
Journal ranking problem $\left( N^{i \cup j},C^{i \cup j} \right) \in \mathcal{J}^{n-1}$ is given by $N^{i \cup j} = \left( N \setminus \{ J_i, J_j \} \right) \cup J_{i \cup j}$ and $C^{i \cup j} = \left[ c^{i \cup j}_{k \ell} \right] \in \mathbb{R}^{(n-1) \times (n-1)}$ such that
\begin{itemize}
\item
$c^{i \cup j}_{k \ell} = c_{k \ell}$ if $\{ J_k, J_\ell \} \cap \{ J_i, J_j \} = \emptyset$;
\item
$c^{i \cup j}_{k (i \cup j)} = c_{ki} + c_{kj}$ for all $J_k \in N \setminus \{ J_i, J_j \}$;
\item
$c^{i \cup j}_{(i \cup j) \ell} = c_{i \ell} + c_{j \ell}$ for all $J_\ell \in N \setminus \{ J_i, J_j \}$.
\end{itemize}
Scoring procedure $f: \mathcal{J}^n \to \mathbb{R}^n$ is called \emph{invariant to aggregation} if $f_k(N,C) \geq f_\ell(N,C)$ implies $f_k \left( N^{i \cup j},C^{i \cup j} \right) \geq f_\ell \left( N^{i \cup j},C^{i \cup j} \right)$ for all $J_k, J_\ell \in N \setminus \{ J_i, J_j \}$.
\end{axiom}

The idea behind invariance to aggregation is that any journal ranking problem can be transformed into a reduced problem by defining the union $J_{i \cup j}$ of journals $J_i$ and $J_j$ as follows: all citations between them are deleted, while any citations by/to these journals are summed up for the ``aggregated'' journal $J_{i \cup j}$. This transformation is required to preserve the order of the journals not affected by the aggregation.

Such an aggregation makes sense, for example, if one is interested only in the ranking of journals from a given field (e.g. economics journals) when journals from other disciplines can be considered as one entity.  

Invariance to aggregation is somewhat related to the \emph{consistency} axiom of \citet{Palacios-HuertaVolij2004}, which is also based on the notion of the reduced problem. However, our property probably takes the information from the missing journal in a more straightforward way into consideration.

Invariance to aggregation has some connections to the famous \emph{independence of irrelevant alternatives} ($IIA$) condition, too, which is used, for example, in Arrow's impossibility theorem \citep{Arrow1951}. Both axioms require an important aspect of the problem, the citations between two journals and the individual preferences between two alternatives, respectively, to remain fixed. However, there is a crucial difference: the set of alternatives (corresponding to journals) is allowed to change in the case of $IA$, while the preferences (corresponding to citations) are allowed to change in the case of $IIA$.
 
\subsection{Self-consistency} \label{Sec32}

This axiom, originally introduced in \citet{ChebotarevShamis1997a} to operators used for aggregating preferences, may require a longer explanation.

First, some reasonable conditions are formulated for the ranking derived from any journal ranking problem. In particular, journal $J_i$ is judged better than journal $J_j$ if one of the following holds:
\begin{enumerate}[label = \pentagon\arabic*), ref = \pentagon\arabic*]
\item \label{SC_con1}
$J_i$ has more favourable citation records against the same journals;
\item \label{SC_con2}
$J_i$ has more favourable citation records against journals with the same quality;
\item \label{SC_con3}
$J_i$ has the same citation records against higher quality journals;
\item \label{SC_con4}
$J_i$ has more favourable citation records against higher quality journals.
\end{enumerate}
In addition, journals $J_i$ and $J_j$ should get the same rank if one of the following holds:
\begin{enumerate}[resume, label = \pentagon\arabic*), ref = \pentagon\arabic*]
\item \label{SC_con5}
they have the same citation records against the same journals;
\item \label{SC_con6}
they have the same citation records against journals with the same quality.
\end{enumerate}

Principles \ref{SC_con2}-\ref{SC_con4} and \ref{SC_con6} can be applied only after measuring the quality of the journals. The name of the property, \emph{self-consistency}, refers to the fact that this is provided by the scoring procedure itself.

The meaning of the requirements above is illustrated by an example.

\begin{figure}[htbp]
\centering
\caption{The journal ranking problem of Example~\ref{Examp31}}
\label{Fig31}
\begin{tikzpicture}[scale=1, auto=center, transform shape, >=triangle 45]
\tikzstyle{every node}=[draw,shape=circle];
  \node (n1) at (135:2) {$J_1$};
  \node (n2) at (45:2)  {$J_2$};
  \node (n3) at (225:2) {$J_3$};
  \node (n4) at (315:2) {$J_4$};

  \foreach \from/\to in {n1/n2,n1/n3,n2/n4,n3/n4}
    \draw [->] (\from) -- (\to);
\end{tikzpicture}
\end{figure}
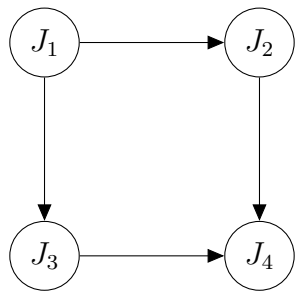

\begin{example} \label{Examp31}
Consider the journal ranking problem $(N,C) \in \mathcal{J}_B^4 \cap \mathcal{J}_U^4 \cap \mathcal{J}_L^4 \cap \mathcal{J}_E^4$ with the following citation matrix:
\[
C = \left[
\begin{array}{cccc}
    0     & 1     & 1     & 0 \\
    0     & 0     & 0     & 1 \\
    0     & 0     & 0     & 1 \\
    0     & 0     & 0     & 0 \\
\end{array}
\right].
\]
This is shown in Figure~\ref{Fig31} where a directed edge from node $J_i$ to $J_j$ indicates that journal $J_i$ has received a citation from journal $J_j$.
\end{example}

Self-consistency has the following implications for the journal ranking problem presented in Example~\ref{Examp31}:
\begin{itemize}
\item
$J_2 \sim J_3$ due to rule~\ref{SC_con5}.
\item
$J_1 \succ J_4$ because of rule~\ref{SC_con1} as $c_{12} > c_{42}$ and $c_{13} > c_{43}$.
\item
Assume for contradiction that $J_1 \preceq J_2$. Then $c_{12} > c_{21}$ and $J_2 \succeq J_1$, as well as $c_{13} = c_{24}$ and $J_3 \sim J_2 \succeq J_1 \succ J_4$, so rule~\ref{SC_con4} leads to $J_1 \succ J_2$, which is impossible.
Consequently, $J_1 \succ (J_2 \sim J_3)$.
\item
Assume for contradiction that $J_2 \preceq J_4$. Then $c_{21} > c_{43}$ and $J_1 \succ J_3$, as well as $c_{24} > c_{43}$ and $J_4 \succeq J_2 \sim J_3$, so rule~\ref{SC_con4} leads to $J_2 \succ J_4$, which is impossible.
Consequently, $(J_2 \sim J_3) \succ J_4$.
\end{itemize}
To conclude, self-consistency demands the ranking to be $J_1 \succ (J_2 \sim J_3) \succ J_4$ in Example~\ref{Examp31}.

It is clear that self-consistency does not guarantee the uniqueness of the ranking in general \citep{Csato2019d}.

Now we turn to the mathematical formulation of this axiom.

\begin{definition} \label{Def31}
\emph{Competitor set}:
Let $(N,C) \in \mathcal{J}_U^n$ be an unweighted journal ranking problem. The \emph{competitor set} of journal $J_i$ is $S_i = \{ J_j: m_{ij} = 1 \}$.
\end{definition}

Journals in the competitor set $S_i$ are called the \emph{competitors} of $J_i$.
Note that $|S_i| = |S_j|$ for all $J_i, J_j \in N$ if and only if the ranking problem is balanced.


The competitor set is defined only for unweighted journal ranking problem but self-consistency may have implications for any pair of journals which have the same number of matches. The generalisation is based on a decomposition of journal ranking problems.

\begin{definition} \label{Def32}
\emph{Sum of journal ranking problems}:
Let $(N,C),(N,C') \in \mathcal{J}^n$ be two journal ranking problems with the same set of journals $N$. The \emph{sum} of these journal ranking problems is the journal ranking problem $(N,C+C') \in \mathcal{J}^n$.
\end{definition}

The sum of journal ranking problems has a number of reasonable interpretations. For instance, they can reflect the citations from different years, or by authors from different countries.

According to Definition~\ref{Def32}, any journal ranking problem can be derived as the sum of unweighted journal ranking problems. However, it might have a number of possible decompositions.

\begin{notation} \label{Not32}
Let $(N,C^{(p)}) \in \mathcal{J}_U^n$ be an unweighted journal ranking problem.
The competitor set of journal $J_i$ is $S_i^{(p)}$.
Let $J_i, J_j \in N$ be two different journals and $g^{(p)}: S_i^{(p)} \leftrightarrow S_j^{(p)}$ be a one-to-one correspondence between the competitors of $J_i$ and $J_j$.
Then $\mathfrak{g}^{(p)}: \{k: J_k \in S_i^{(p)} \} \leftrightarrow \{\ell: J_\ell \in S_j^{(p)} \}$ is given by $J_{\mathfrak{g}^{(p)}(k)} = g^{(p)}(J_k)$.
\end{notation}

Finally, we are able to introduce conditions~\ref{SC_con1}-\ref{SC_con6} with mathematical formulas.

\begin{axiom} \label{Axiom32}
\emph{Self-consistency} ($SC$) \citep{ChebotarevShamis1997a}:
Scoring procedure $f: \mathcal{J}^n \to \mathbb{R}^n$ is called \emph{self-consistent} if the following implication holds for any journal ranking problem $(N,C) \in \mathcal{J}^n$ and for any journals $J_i,J_j \in N$:
if there exists a decomposition of the journal ranking problem $(N,C)$ into $m$ unweighted journal ranking problems -- that is, $C = \sum_{p=1}^m C^{(p)}$ and $(N,C^{(p)}) \in \mathcal{J}_U^n$ is an unweighted journal ranking problem for all $p = 1,2, \dots ,m$ -- together with the existence of a one-to-one mapping $g^{(p)}$ from $S^{(p)}_i$ onto $S^{(p)}_j$ such that $c_{ik}^{(p)} \geq c_{j \mathfrak{g}^{(p)}(k)}^{(p)}$ and $f_k(N,C) \geq f_{\mathfrak{g}^{(p)}(k)}(N,C)$ for all $p = 1,2, \dots ,m$ and $J_k \in S_i^{(p)}$, then
$f_i(N,C) \geq f_{j}(N,C)$.
Furthermore, $f_i(N,C) > f_{j}(N,C)$ if $c_{ik}^{(p)} > c_{j \mathfrak{g}^{(p)}(k)}^{(p)}$ or $f_k(N,C) > f_{\mathfrak{g}^{(p)}(k)}(N,C)$ for at least one $1 \leq p \leq m$ and $J_k \in S_i^{(p)}$.
\end{axiom}

In a nutshell, self-consistency implies that if journal $J_i$ does not show worse performance than journal $J_j$ on the basis of the citation matrix, then it is not ranked lower, in addition, it is ranked strictly higher when it becomes clearly better.

\citet{ChebotarevShamis1997a} consider only loopless journal ranking problems but the extension of self-consistency is trivial as presented above.


\citet[Theorem~5]{ChebotarevShamis1998a} gives a necessary and sufficient condition for self-consistent scoring procedures, while \citet[Table~2]{ChebotarevShamis1998a} presents some scoring procedures that satisfy this requirement.
See also \citet{Csato2019d} for an extensive discussion of self-consistency.

\section{The incompatibility of the two axioms} \label{Sec4}

In the following, it will be proved that no scoring procedure can meet axioms $IA$ and $SC$.

\begin{figure}[htbp]
\centering
\caption{The journal ranking problems of Example~\ref{Examp41}}
\label{Fig41}
  
\begin{subfigure}{.5\textwidth}
  \centering
  \subcaption{Journal ranking problem $(N,C)$}
  \label{Fig41a}
\begin{tikzpicture}[scale=1, auto=center, transform shape, >=triangle 45]
\tikzstyle{every node}=[draw,shape=circle]; 
  \node (n1) at (135:2) {$J_1$};
  \node (n2) at (45:2)  {$J_2$};
  \node (n3) at (225:2) {$J_3$};
  \node (n4) at (315:2) {$J_4$};

  \draw (n1) -- (n2);
  \draw (n1) -- (n3);
  \draw (n2) -- (n4);
  \draw [->] (n3) -- (n4);
\end{tikzpicture}
\end{subfigure}
\begin{subfigure}{.5\textwidth}
  \centering
  \subcaption{Journal ranking problem $(N^{3 \cup 4},C^{3 \cup 4})$}
  \label{Fig41b}
\begin{tikzpicture}[scale=1, auto=center, transform shape, >=triangle 45]
\tikzstyle{every node}=[draw,shape=circle];
  \node (n1) at (135:2) {$J_1$};
  \node (n2) at (45:2)  {$J_2$};
  \node (n3) at (270:1.35) {$J_{3 \cup 4}$};

  \draw (n1) -- (n2);
  \draw (n1) -- (n3);
  \draw (n2) -- (n3);
\end{tikzpicture}
\end{subfigure}
\end{figure}
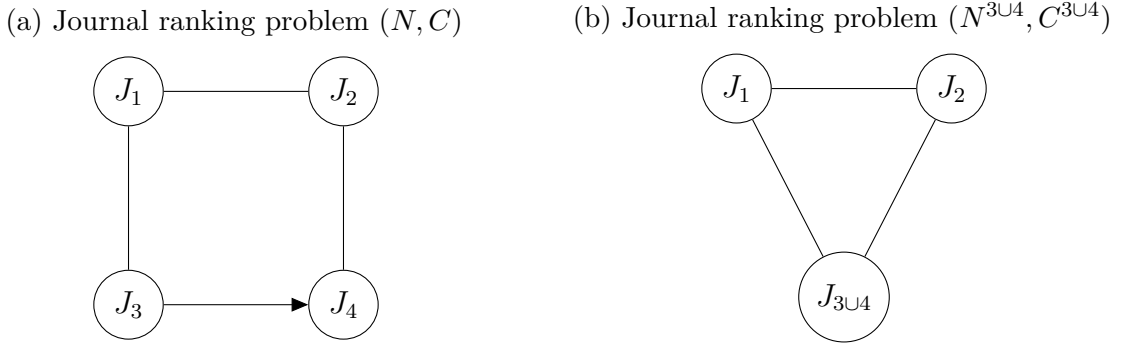

\begin{example} \label{Examp41}
Let $(N,C) \in \mathcal{J}_B^4 \cap \mathcal{J}_U^4 \cap \mathcal{J}_L^4 \cap \mathcal{J}_E^4$ and $(N^{3 \cup 4},C^{3 \cup 4}) \in \mathcal{J}_B^3 \cup \mathcal{J}_U^3 \cap \mathcal{J}_L^3 \cap \mathcal{J}_E^3$ be the journal ranking problems with the citation matrices
\[
C = \left[
\begin{array}{K{1.5em}K{1.5em}K{1.5em}K{1.5em}}
    0     & 0.5   & 0.5   & 0   \\
    0.5   & 0     & 0     & 0.5 \\
    0.5   & 0     & 0     & 1   \\
    0     & 0.5   & 0     & 0   \\
\end{array}
\right] \qquad \text{and} \qquad
C^{3 \cup 4} = \left[
\begin{array}{K{1.5em}K{1.5em}K{1.5em}}
    0     & 0.5   & 0.5 \\
    0.5   & 0     & 0.5 \\
    0.5   & 0.5   & 0 \\
\end{array}
\right] \text{, respectively}.
\]
Journal ranking problem $(N^{3 \cup 4},C^{3 \cup 4})$ is obtained by uniting journals $3$ and $4$.

This is shown in Figure~\ref{Fig41} where a directed edge from node $J_i$ to $J_j$ indicates that journal $J_i$ has received a citation from journal $J_j$, and an undirected edge between the nodes means that the two journals are tied by mutual citations.
\end{example}

\begin{theorem} \label{Theo41}
There exists no scoring procedure that is invariant to aggregation and self-consistent.
\end{theorem}

\begin{proof}
The contradiction of the two properties can be proved by Example~\ref{Examp41}.
Take first the journal ranking problem $(N,C)$, which has the competitor sets $S_1 = S_4 = \{ J_2, J_3 \}$ and $S_2 = S_3 = \{ J_1, J_4 \}$.
Assume for contradiction the existence of a scoring procedure $f: \mathcal{J}^n \to \mathbb{R}^n$ satisfying invariance to aggregation and self-consistency.

Self-consistency has several implications for the scoring procedure $f$ as follows:
\begin{enumerate}[label=\emph{\alph*})]
\item \label{Enum_a}
Consider the (identity) one-to-one correspondence $g_{14}: S_1 \leftrightarrow S_4$ such that $g_{14}(J_2) = J_2$ and $g_{14}(J_3) = J_3$. Then $g_{14}$ satisfies condition~\ref{SC_con1} of $SC$ due to $c_{12} = c_{42} = 0.5$ and $0.5 = c_{13} > c_{43} = 0$, thus $f_1(N,C) > f_4(N,C)$.

\item \label{Enum_b}
Consider the (identity) one-to-one correspondence $g_{32}: S_3 \leftrightarrow S_2$ such that $g_{32}(J_1) = J_1$ and $g_{32}(J_4) = J_4$. Then $g_{32}$ satisfies condition~\ref{SC_con1} of $SC$ due to $c_{31} = c_{31} = 0.5$ and $1 = c_{34} > c_{24} = 0.5$, thus $f_3(N,C) > f_2(N,C)$.

\item \label{Enum_c}
Suppose that $f_2(N,C) \geq f_1(N,C)$, which implies $f_3(N,C) > f_4(N,C)$ according to the inequalities derived in \ref{Enum_a} and \ref{Enum_b}.
Consider the one-to-one correspondence $g_{12}: S_1 \leftrightarrow S_2$ such that $g_{12}(J_2) = J_1$ and $g_{12}(J_3) = J_4$. Then $g_{12}$ satisfies condition~\ref{SC_con3} of $SC$ due to $c_{12} = c_{21} = 0.5$ and $c_{13} = c_{24} = 0.5$, thus $f_1(N,C) > f_2(N,C)$, a contradiction.
\end{enumerate}

Therefore $f_1(N,C) > f_2(N,C)$ should hold, when invariance to aggregation results in $f_1(N,C') > f_2(N,C')$.
However, self-consistency leads to $f_1(N,C') = f_2(N,C')$ in the journal ranking problem $(N,C')$ because of the one-to-one mapping $g_{12}: S_1' \leftrightarrow S_2'$ such that $g_{12}(J_2) = J_1$ and $g_{12}(J_1) = J_2$: the assumption of $f_1(N,C') > f_2(N,C')$ implies $f_1(N,C') < f_2(N,C')$ due to condition~\ref{SC_con3} (the competitors of $J_2$ are more prestigious), while $f_1(N,C') < f_2(N,C')$ would result in $f_1(N,C') < f_2(N,C')$ due to condition~\ref{SC_con3} (the competitors of $J_1$ are more prestigious) again.

Hence a scoring procedure cannot meet $IA$ and $SC$ at the same time.
\end{proof}

Since Example~\ref{Examp41} contains balanced, unweighted, loopless, and extremal journal ranking problems, there is few hope to avoid the impossibility of Theorem~\ref{Theo41} by plausible domain restrictions.

\begin{remark} \label{Rem41}
$IA$ and $SC$ are logically independent axioms as there exist scoring procedures that satisfy one of the two properties: the least squares method is self-consistent \citep[Theorem~5]{ChebotarevShamis1998a}, and the flat scoring procedure, which gives $f_i(N,C) = 0$ for all $i \in N$ and $(N,C) \in \mathcal{J}^n$, is invariant to aggregation.
\end{remark}

\section{Conclusions} \label{Sec5}

We have presented an impossibility theorem in journal ranking: a reasonable method cannot be invariant to the aggregation of journals, even in the case of a substantially restricted domain of citation graphs. An intuitive explanation is that invariance to aggregation is a local property (it modifies only the citations directly affecting the journals to be united), while self-consistency considers the global structure of the citations as it depends on the quality of the journals. The clash between local and global axioms can also be observed in other fields, such as game theory. In addition, the impossibility result clearly shows that invariance to aggregation is a rather strong requirement, similarly to its peer independence of irrelevant alternatives \citep{MalawskiZhou1994}.
Nevertheless, according to our finding, the choice of the set of journals to be compared is an important aspect of every empirical study which aims to measure intellectual influence.

It is clear that the axiomatic analysis discussed here has a number of limitations as it is able to consider indices from only one point of view \citep{GlanzelMoed2013}. For example, the citation graph is assumed to be known, that is, the issue of choosing an adequate time window is neglected. In addition, this paper has not addressed several important problems of scientometrics such as the comparability of distant research areas, or the proper treatment of different types of publications.

To summarise, the derivation of similar impossibility results may contribute to a better understanding of the inevitable trade-offs between various properties, and it means a natural subject of further studies besides axiomatic characterisations.

\section*{Acknowledgements}
\addcontentsline{toc}{section}{Acknowledgements}
\noindent
We are grateful to \emph{Gy\"orgy Moln\'ar} and \href{https://sites.google.com/view/doragretapetroczy}{\emph{D\'ora Gr\'eta Petr\'oczy}} for inspiration. \\
Two anonymous reviewers provided valuable comments and suggestions on an earlier draft. \\
The research was supported by OTKA grant K 111797 and by the MTA Premium Postdoctoral Research Program.

\bibliographystyle{apalike}
\bibliography{All_references}

\end{document}